\documentclass[a4paper,11pt]{article}
\usepackage[margin=1in]{geometry}

\usepackage[utf8]{inputenc}
\usepackage[T1]{fontenc}
\usepackage{amsmath,amsthm,amssymb,thmtools,thm-restate,booktabs,color,doi,graphicx,latexsym,url,xcolor,xspace}
\usepackage[numbers,sort&compress]{natbib}
\usepackage{microtype,hyperref}
\usepackage[inline]{enumitem}
\setlist[itemize]{label=--}
\setlist[enumerate]{label=(\arabic*),labelindent=\parindent,leftmargin=*}

\usepackage{caption}
\usepackage{subcaption}
\usepackage{authblk}

\usepackage{tikz,pgfplots}
\usepackage{siunitx}
\usepackage{pgfplotstable,colortbl}
\usetikzlibrary{arrows,shapes,backgrounds,positioning,fit,snakes}
\pgfmathsetseed{420}

\definecolor{citecolor}{HTML}{0000C0}
\definecolor{urlcolor}{HTML}{000080}

\hypersetup{
    colorlinks=true,
    linkcolor=black,
    citecolor=citecolor,
    filecolor=black,
    urlcolor=urlcolor,
    pdftitle={}, 
    pdfauthor={} 
}

\newtheorem{theorem}{Theorem}
\newtheorem{lemma}[theorem]{Lemma}

\newtheorem{fact}[theorem]{Fact}

\newcommand{\Gsigma}{\mathcal{G}_\sigma}

\newcommand\janne[2][]{}
\newcommand\stefan[2]{}
\newcommand\nnote[2][]{}
\newcommand\alex[2]{}

\newcommand{\avgdeg}{\Delta_{\text{avg}}}

\newcommand{\demandbalancingalg}{\textsc{Demand Balancing}}
\newcommand{\optimizedthreshalg}{\textsc{Threshold Balancing}}
\newcommand{\sqrtalg}{\textsc{Split}\&\textsc{Join}}
\newcommand{\helperalg}{\textsc{Helper}}
\newcommand{\fixeddegalg}{\textsc{Fixed-degree}}
\newcommand{\steinernodealg}{\textsc{Steiner Node Insertion}}
\newcommand{\randtreealg}{\textsc{Random Tree}}
\newcommand{\randgraphalg}{\textsc{Random Graph}}
\newcommand{\greedyselectalg}{\textsc{Greedy Edge Selection}}
\newcommand{\greedydeletealg}{\textsc{Greedy Edge Deletion}}
\newcommand{\hybriddeletealg}{\textsc{Hybrid Edge Deletion}}

\newcommand{\demandbalancingalgshort}{\textsc{DB}}
\newcommand{\optimizedthreshalgshort}{\textsc{TB}}
\newcommand{\sqrtalgshort}{\textsc{S}\&\textsc{J}}
\newcommand{\helperalgshort}{\textsc{Help}}
\newcommand{\fixeddegalgshort}{\textsc{Fixed-deg.}}
\newcommand{\steinernodealgshort}{\textsc{SNI}}
\newcommand{\randtreealgshort}{\textsc{Rnd. Tree}}
\newcommand{\randgraphalgshort}{\textsc{Rnd. Graph}}
\newcommand{\greedyselectalgshort}{\textsc{GES}}
\newcommand{\greedydeletealgshort}{\textsc{GED}}
\newcommand{\hybriddeletealgshort}{\textsc{HED}}

\begin{document}

\title{Demand-Aware Network Design with Steiner Nodes\\and a Connection to Virtual Network Embedding}

\author[1]{Aleksander Figiel}
\affil[1]{TU Berlin, Germany}
\author[1]{Janne H.\ Korhonen}
\author[2]{Neil Olver}
\affil[2]{London School of Economics and Political Science, United Kingdom}
\author[1]{Stefan Schmid}

\date{}

\maketitle

\begin{abstract}
Emerging optical and virtualization technologies enable the design of more flexible and demand-aware networked systems, in which resources can be optimized toward the actual workload they serve. For example, in a demand-aware datacenter network, frequently communicating nodes (e.g., two virtual machines or a pair of racks in a datacenter) can be placed topologically closer, reducing communication costs and hence improving the overall network performance.

This paper revisits the \emph{bounded-degree network design problem} underlying such demand-aware networks. Namely, given a distribution over communicating server pairs, we want to design a network with bounded maximum degree that minimizes expected communication distance. In addition to this known problem, we introduce and study a variant where we allow Steiner nodes (i.e., additional routers) to be added to augment the network.

We improve the understanding of this problem domain in several ways. First, we shed light on the complexity and hardness of the aforementioned problems, and study a connection between them and the virtual networking embedding problem. We then provide a constant-factor approximation algorithm for the Steiner node version of the problem, and use it to improve over prior state-of-the-art algorithms for the original version of the problem with sparse communication distributions. Finally, we investigate various heuristic approaches to bounded-degree network design problem, in particular providing a reliable heuristic algorithm with good experimental performance.

We report on an extensive empirical evaluation, using several real-world traffic traces from datacenters, and find that our approach results in improved demand-aware network designs.
\end{abstract}

\makeatletter
\def\blfootnote{\gdef\@thefnmark{}\@footnotetext}
\makeatother
\blfootnote{This project has received funding from the European Research Council (ERC) under the European Union’s Horizon 2020 research and innovation programme (Grant agreement No. 864228 “AdjustNet”).}

\section{Introduction}

With the popularity of datacentric applications and artificial intelligence, the traffic in datacenters is growing explosively. Accordingly, over the last years, significant research efforts have been made to render datacenter networks and other distributed computing infrastructure more efficient and flexible~\cite{kellerer2019adaptable}. 

An intriguing approach to improve performance in networked and distributed systems, is to make them more demand-aware. In particular, emerging optical communication technologies allow to dynamically change the topology interconnecting the datacenter racks, according to the workload they currently serve~\cite{poutievski2022jupiter,projector}. 
But also virtualization technologies enable demand-awareness, e.g., by allowing to flexibly embed a workload (e.g., processes or virtual machines) on a given infrastructure~\cite{henzinger2019efficient}.

Such demand-awareness is attractive as empirical studies (e.g., from Google~\cite{poutievski2022jupiter}, Microsoft~\cite{ballani2020sirius}, Facebook~\cite{facebook2013}) show that communication traffic features significant spatial locality, which can be exploited for optimization~\cite{sigmetrics20complexity}. 
For example, in a demand-aware network, frequently communicating nodes (e.g., a pair of racks in a datacenter) can be placed topologically closer, reducing communication costs and hence improving the overall performance.  
However, the underlying algorithmic problem of such demand-aware networks is still not well understood today.

This paper revisits the design of efficient demand-aware networks, and in particular the bounded-degree network topology design problem~\cite{avin2020demand}: given a traffic matrix describing the communication demand between pairs of nodes, find a network topology (i.e., a graph) which ``serves this traffic well.'' In particular and informally, the network interconnecting the nodes should be demand-aware in the sense that it provides short paths between frequently communicating node pairs, minimizing the number of hops travelled per bit. The degree bound is motivated by practical constraints and scalability requirements: reconfigurable switches typically have a limited number of ports (or lasers/mirrors)~\cite{osn21}. 

Before we present our results in more detail, let us introduce our model more formally.


\paragraph{Bounded network design.}  Formally, in the \emph{bounded-degree demand-aware network design problem} (or \emph{bounded network design problem} for short), we are given a degree bound $\Delta$, a node set $V = \{ 1, 2, \dotsc, n  \}$ and a probability distribution $p \colon \binom{V}{2} \to [0,1]$ over possible edges over set $V$. The task is to find a graph $G = (V, E)$ with maximum degree $\Delta$ minimizing the \emph{expected path length}
\begin{equation}\label{eq:epl} \sum_{ \{u, v\} \in \binom{V}{2}} p(\{u,v\}) d_G(u,v)\,, 
\end{equation}
where $d_G(u,v)$ denotes the distance between $u$ and $v$ in $G$. We interpret the edges $e \in \binom{V}{2}$ with non-zero probability $p(e)$ as forming a graph $D = (V, E_D)$ that we will refer to as a \emph{demand graph} or \emph{guest graph}. We call the edges of the demand graph \emph{demand edges}. The graph obtained as a solution is called the \emph{host graph}.

We note that compared to the original formulation of the problem given by Avin et al.~\cite{disc17,avin2020demand}, we define the demand distribution over undirected edges instead of directed edges. However, one may easily observe that the two formulations are equivalent: adding the demands on two opposite directed edges to a single undirected edge or splitting the demand of an undirected edge into equal demands on the two corresponding directed edges results in the same objective function.

The state-of-the-art bounded network design solutions~\cite{disc17,avin2020demand} come with the drawback that a constant approximation is only achieved for sparse demand graphs and that the solutions can violate the degree bound $\Delta$ to some extent. The complexity of the general problem and the achievable expected path length are still not well understood. 

\paragraph{Bounded network design with Steiner nodes.}  In this work, we introduce and study a generalization of the bounded network design problem, in a resource augmentation model where we are allowed to use \emph{Steiner nodes} in addition to the nodes of the guest graph. The problem may be of independent interest, but it also turns out to be particularly useful to overcome some of the limitations of existing approaches to the original bounded network design problem. 
Our generalization introduces both an algorithmic and an analytical tool, which will allow us to obtain improved results also for the original model. 
Formally, in the \emph{bounded network design problem with Steiner nodes}, we are given a degree bound $\Delta$, a node set $V = \{ 1, 2, \dotsc, n  \}$ and a probability distribution $p \colon \binom{V}{2} \to [0,1]$ over possible edges over set $V$ as before. The task is to find a graph $G = (V \cup U, E)$ with maximum degree $\Delta$, where $U$ is an arbitrary set disjoint from $V$, minimizing the expected path length \eqref{eq:epl}. 
We refer to nodes in $U$ as \emph{Steiner nodes}. 
For sanity, one may observe that there is always an optimal solution with polynomial number of Steiner nodes (see Lemma~\ref{lemma:steiner-node-bound}).

We show that this variant is NP-hard and can be approximated within factor $2 + o(1)$ in polynomial time. The approximation algorithm will further be useful in designing various heuristic algorithms for the standard version of bounded network design, as we discuss below.

\paragraph{Our contributions.} 
We revisit the design of demand-aware networks from an optimization perspective. 
We first establish an intriguing connection between the problem of designing demand-aware network topologies and the virtual network embedding problem, where the guest graph can be placed on a given (demand-oblivious) network in a demand-aware manner. In particular, we show that the latter does not provide good approximations for the topology design problem: given any  bounded-degree network (e.g., an expander graph) chosen in a demand-oblivious manner, there exist demands for which the optimal embedding into this network costs an $\Omega(\log\log n)$ factor more than the optimal demand-aware embedding, for networks of size $n$

We then aim to overcome the limitations of existing bounded-degree demand-aware network topologies design approaches, and introduce the bounded network design with Steiner nodes problem. We present a constant factor approximation algorithm for this problem, 
and show how it can be used in the context of the original bounded network design problem as well as to design heuristics which overcome the limitations of the state of the art. We further rigorously prove both problems with and without Steiner nodes are NP-hard.

We further report on an extensive empirical evaluation, using several real-world traffic traces from datacenters. We find that our approach is attractive in these scenarios as well, when compared to the state of the art, both in terms of the achieved expected path length and the degree bound. 

\paragraph{Related work.} 
Demand-aware networks are motivated by their applications in datacenters, where reconfigurable optical communication technologies recently introduced unprecedented flexibilities~\cite{projector,zhang2021gemini,apocs21renets,sigmetrics22cerberus,zerwas2023duo,firefly,fleet,flexspander}.
Demand-aware networks have also been studied in the context of virtualized infrastructures, and especially the virtual network embedding problem has been subject to intensive research over the last decade~\cite{rost2019virtual,figiel2021optimal,rost2019virtual,chowdhury2009virtual,henzinger2019efficient}.

The bounded-degree demand-aware network design problem
was introduced by Avin et al.~at DISC 2017 \cite{disc17,avin2020demand}.
The authors show that the problem is related to demand-aware data structures, such as biased binary search trees \cite{splaytrees}, and provide a constant-approximation algorithm for sparse demands in a constant degree-augmentation model. That  paper builds upon ideas from  SplayNets~\cite{schmid2015splaynet} and there already exists much follow-up work, e.g., on designs accounting for robustness~\cite{avin2018rdan},  congestion~\cite{avin2022demand}, or scenarios where the demand-aware network is built upon a demand-oblivious network~\cite{avin2021renets}, among many other~\cite{hall2021survey}.
In this paper, we improve upon the bounds of \cite{disc17,avin2020demand}, even in a more general model. We are also not aware on any algorithm engineering approaches to render these algorithms more efficient in practice.

\paragraph{Organization.} 
The remainder of this paper is organized as follows.
After introducing preliminaries in Section~\ref{sec:preliminaries},
we present hardness results and lower bounds in Section~\ref{sec:complexity}.
Our approximation approach with Steiner nodes is presented in Section~\ref{sec:steiner-apx}. In 
Section~\ref{sec:demand-balancing-model} we discuss applications of the Steiner node approach for the original bounded network design problem without Steiner nodes.
We report on our empirical results from experiments on datacenter traffic traces in Section~\ref{sec:experiments} and conclude in
Section~\ref{sec:conclusion}.


\section{Preliminaries}\label{sec:preliminaries}

\paragraph{Basic notation.} Let $V$ be a node set and $p \colon \binom{V}{2} \to [0,1]$ a probability distribution given as input for the bounded network design problem. For nodes $v, u \in V$, we define
\[
p(v)  = \sum_{u \in V} p(\{v,u\})\,, \hspace{20mm}  p_v(u)  = p(\{v,u\} )/p(v)\,.
\]


\paragraph{Entropy.} Let $X$ be a random variable with a finite domain $\mathcal{X}$. Recall that the \emph{entropy} and the \emph{base-$d$ entropy} of $X$, respectively, are 
\[ H(X) = -\sum_{x \in \mathcal{X}} \Pr(X = x) \log_2 \Pr(X = x)\quad\text{and}\quad H_d(X) = -\sum_{x \in \mathcal{X}} \Pr(X = x) \log_d \Pr(X = x)\,.\]
For any $d$, we have $H_d(X) = H(X)/\log_2 d$. When $X$ is distributed according to $p \colon \mathcal{X} \to [0,1]$, we abuse the notation by writing $H(p)$ for the entropy of $X$.

For random variables $X$ and $Y$ with finite domains $\mathcal{X}$ and $\mathcal{Y}$, respectively, the (base-$d$) \emph{entropy of $Y$ conditioned on $X$} is
\[H_d(Y \mid X) = - \sum_{x \in \mathcal{X}} \sum_{y \in \mathcal{Y}} \Pr\bigl((X,Y) = (x,y)\bigr) \log_d \Pr(Y = y \mid X = x)\,.\]

\paragraph{Entropy bound on trees.} Consider a setting where we have a root node $r$, a set of nodes $V$, and demands between the root and nodes in $V$ given by a probability distribution $p \colon V \to [0,1]$, i.e., demand between $r$ and $v \in V$ is $p(v)$. Suppose we want a tree rooted at $r$ and including all nodes of $V$ such that the expected path length between $r$ and $V$, i.e., $\sum_{v \in V} p( v) d_T(r,v)$ is minimized. It is known that this is closely related to the entropy $H(p)$.

\begin{lemma}[\cite{avin2020demand,mehlhorn1975nearly}]
Let $\{ r \} \cup V$ be a set of nodes and let $p \colon V \to [0,1]$ be a probability distribution. Then for any $d$-ary tree with root $r$ and node set including $\{ r \} \cup V$, it holds that
\[ \sum_{v \in V} p( v) d_T(r,v) \ge H_{d+1}(p) - 1\,. \]
\end{lemma}

Note that the above lemma covers cases where nodes in $V$ can appear as either internal nodes or as leaves. If nodes in $V$ are required to appear as leaves only, then this question is equivalent to asking for optimal \emph{prefix codes}.

\begin{lemma}[Huffman coding, e.g.~\cite{2006elements,alistair-huffman}]\label{lemma:huffman}
Let $\{ r \} \cup V$ be a set of nodes with $|V| = n$ and let $p \colon V \to [0,1]$ be a probability distribution. The $d$-ary tree $T$ with root $r$ with leaf set exactly $V$ minimizing the expected path length $\sum_{v \in V} p( v) d_T(r,v)$ can be computed in time $O(n \log n)$ and satisfies
\[ H_d(p) \le \sum_{v \in V} p( v) d_T(r,v) \le H_d(p) + 1\,. \]
\end{lemma}

\paragraph{Entropy bound on graphs.} For the bounded network design problem, both with and without Steiner nodes, we use the following result to lower bound the cost of an optimal solution.

\begin{theorem}[\cite{avin2020demand}]\label{thm:lower-bound}
Assume we have a degree bound $\Delta$, a node set $V = \{ 1, 2, \dotsc, n  \}$ and a probability distribution $p \colon \binom{V}{2} \to [0,1]$ over possible edges over set $V$. Then for any graph $G = (V \cup U, E)$ with maximum degree $\Delta$, the expected path length between nodes of $V$ satisfies
\[ \sum_{u,v \in V} p(\{u,v\}) d_G(u,v) \ge \frac{1}{2}\sum_{v \in V} p(v) H_{\Delta+1}(p_v) - 1\,.\]
\end{theorem}

Note that while there appears to be a factor $\frac{1}{2}$ in the above inequality that is not present in the original formulation of Avin et al.~\cite{avin2020demand}; this is simply due to the fact that we define the demand distribution over undirected edges rather than directed edges. \nnote{Do we need the rest of this? Or can it go in the appendix?} To see how this arises, consider a distribution over directed edges where each edge $(u,v) \in V \times V$ has probability $\frac{1}{2}p(\{ u,v \})$.

The quantity 
\[ \frac{1}{2}\sum_{v \in V} p(v) H_{\Delta+1}(p_v) \]
can be interpreted as a conditional entropy as follows. Let $X$ and $Y$ be random variables over $V$ defined by picking $X$ to be node $v \in V$ with probability $\frac{1}{2}p(v)$, and then picking $Y$ to be a node $u \in V$ with probability $p_v(u)$. By direct calculation one can see that $\Pr( \{ X,Y \} = \{ u, v \}) = p(\{ u,v \})$. Moreover, the entropy of $Y$ conditioned on $X$ is
\[ H_d(Y \mid X) = \frac{1}{2}\sum_{v \in V} p(v) H_{d}(p_v)\,.\] 

\section{Hardness and lower bounds}
\label{sec:complexity}

\subsection{NP-completeness}
Consider the bounded network problem with integer weights in its decision form: is there a solution with expected path length at most some given value $K$?
We show hardness for both the versions with and without Steiner nodes.
\begin{theorem}\label{thm:np-completeness-main}
Bounded network design and bounded network design with Steiner nodes are NP-complete for any fixed $\Delta \ge 2$.
\end{theorem}
We describe all the details in Appendix~\ref{app:hardness}.
Here, we briefly comment on some main ideas.

A first detail is that we verify that the version with Steiner nodes is in fact in NP; the subtlety here is that we need to show that a polynomial number of Steiner nodes suffices for an optimal solution.

For $\Delta=2$, Steiner nodes are clearly not helpful, and so both variants are the same.
NP-completeness follows easily from a variant of the minimum linear arrangement problem, namely \emph{undirected circular arrangement}.
Here, the goal is to embed the nodes of a graph onto the discrete circle, with minimum average stretch.
If the demand graph is connected, this is identical to the bounded network problem; handling demand graphs with multiple components requires some minor technical work.

For $\Delta > 3$, we have to work harder.
We give a reduction from vertex cover on 3-regular graphs.
There are two main tricks we utilize in this proof.
First is the observation that when we construct instances for either variant of the bounded network problem, we can force any optimal host graph to have any specific edge $\{ u, v \}$ we desire, by giving this pair a very large demand $p(\{ u,v \}) = W$. 
For $W$ large enough compared to other demands, any optimal host graph will necessarily include all of these large-demand edges.

The second trick is that we can essentially allow for each node $v$ to be given its own degree bound $\Delta_v \leq \Delta$. 
This is very useful in the construction of various gadgets used for the reduction. 
To achieve this, we describe ``degree blocking gadgets'': a graph with all nodes having degree $\Delta$, except for one node of degree $\Delta-1$. 
(This requires $\Delta$ to be odd, since no such gadget exists for even $\Delta$ for parity reasons; we have an additional reduction from even $\Delta$ to odd $\Delta$.)
By attaching such a gadget with forced edges (as per the above) to some node $v$, we reduce its degree available for other uses by 1. 

Given an instance of vertex cover on a graph $G=(V,E)$ and a bound $k$ on the size of a vertex cover, an instance is constructed with the following properties.
Aside from forced edges, the demands are all unit, between a root node $r$ and vertices $t_e$, one for each edge $e \in E$.
With the aid of forced edges and reduced degree bounds, the instance is constructed so that every path from $r$ to $t_e$ can be shorter than some minimum $K$; indeed these paths consist of $K-1$ forced edges, and can only be of length $K$ if certain specific edges are included in the host graph.
Further, the only way to have all these paths having length precisely $K$ is through a solution to the vertex cover problem of size at most $k$.
Steiner nodes are not useful for obtaining a solution with expected path length $K$, so this hardness reduction applies to both variants of the model.


\subsection{A lower bound for a universal host graph}\label{sec:universal-lb}

We now return to the question of approximation algorithms for bounded network design \emph{without} Steiner nodes. One possible approach one might consider is choosing a ``universal'' host graph of maximum degree $\Delta$, e.g., a $\Delta$-regular expander, and the solving the problem by embedding the demands into this universal host graph. We show in this section that this approach cannot yield a constant-factor approximation algorithm.

We will restrict our attention to uniform demand functions, and will henceforth simply refer to demand graphs with the implicit understanding that the demand function is uniform.
The precise question we now want to ask is whether there exists a host graph $G$ such that the following holds: given a demand graph $D$ with $n$ vertices, which has an optimal embedding into a host graph $G^*(D)$ (with maximum degree $\Delta$) of expected path length $\ell^*(D)$, we have an embedding into $G$ of expected path length $O(\ell^*(D))$.

\begin{theorem}
Let $\Delta \ge 2$. For any host graph $G=(V,E)$ of size $n$ and maximum degree $\Delta$, there exists some demand graph $D$ with maximum degree $\Delta$ for which the optimal embedding of $D$ into $G$ has expected path length $\Omega(\log \log n)$, where constants depending on $\Delta$ are hidden in the $\Omega$.
\end{theorem}

Note that such a demand graph $D$ trivially has an embedding into a graph of maximum degree $\Delta$ and with expected path length 1; simply use $D$ itself.

\begin{proof}
We consider $\Delta \geq 3$; the claim is obvious for $\Delta=2$, as the only connected host graphs are the $n$-path and the $n$-cycle, and a demand graph consisting of two cycles of length $n/2$ does not have a good embedding into either.

Let $\mathcal{G}_\Delta$ be the collection of $\Delta$-regular graphs on $V$.
We assume from hereon that $\Delta n$ is even, since otherwise $\mathcal{G}_{\Delta}$ is empty.
Standard results allow us to bound $|\mathcal{G}_\Delta|$. 
Asymptotically, the number of $\Delta$-regular graphs (ignoring constants depending on $\Delta$) is (see~\cite{Bender_Canfield_1978})
\[
   \Theta\left(\frac{(\Delta n)!}{(\Delta n /2)!\,2^{\Delta n/2}(\Delta !)^n}\right).
\]
This gives us the lower bound
\[
|\mathcal{G}_\Delta| = \Omega((\Delta n/4)^{\Delta n/2}\cdot \Delta^{-\Delta n}) = \Omega(2^{\Delta n(\log(\Delta n) - 2)/2 - \Delta \log \Delta \cdot n}). 
\]

Suppose that for every graph $D$ on $V$ with maximum degree $\Delta$, there is some embedding (that is, a permutation) $\pi_D$ that embeds $D$ into $G$ with average stretch at most $\ell := \log \log n / (200  \log \Delta)$.
Let $\mathcal{G}_\pi := \{ D \in \mathcal{G}_\Delta: \pi_D = \pi\}$, for any permutation $\pi$ of $V$.
There must be some permutation $\sigma$ for which 
\[
   |\Gsigma| \geq |\mathcal{G}_\Delta| / n! = \Omega(2^{\Delta n\log n / 2 - \Delta (1 + \log \Delta) n - n\log n}).
\]
As long as $\Delta \geq 3$, this is at least $\Omega(2^{\Delta n\log n / 4})$. 

We will restrict our attention to $\Gsigma$ going forward.
Without loss of generality, we can assume $\sigma$ is the identity.

Let $\hat{G} = (V, \hat{E})$ be the graph where $\{ u,v\} \in \hat{E}$ whenever the shortest path distance between $u$ and $v$ in $G$ at most $100\ell$.

\begin{lemma}
$\hat{G}$ has fewer than $2^{100\log \Delta\cdot \ell}\cdot n$ edges.
\end{lemma}
\begin{proof}
For each $v \in V$, just from the degree bound for $G$ we deduce that there are less than $\Delta^{100\ell}$ nodes within distance $100\ell$ of $v$.
\end{proof}
From our choice of $\ell$, this means that $\hat{G}$ has at most $n\sqrt{\log n}$ edges.

We consider every possible subgraph $G'$ of $\hat{G}$.
For each such $G'$, let 
\[ \Gamma(G') := \{ D \in \Gsigma: E(D) \cap \hat{E} = E(G')\}. \]
Since there are at most $2^{n\sqrt{\log n}}$ subgraphs of $\hat{G}$, it follows that there exists some subgraph $G^*$ of $\hat{G}$ for which 
\begin{equation}\label{eq:GammaLB}
    |\Gamma(G^*)| \geq |\Gsigma|/2^{n\sqrt{\log n}} = \Omega(2^{\Delta n\log n / 4 - n\sqrt{\log n}}).
\end{equation}
For each $D \in \Gamma(G^*)$, given that the average stretch of $D$ is at most $\ell$, we have by Markov's inequality that at most $1/100$ fraction of the edges of $D$ have stretch more than $100\ell$. 
That is, $|E(D) \setminus \hat{E}| \leq |E(D)|/100 \leq \Delta n / 200$. 
This means we can describe $D$ in an efficient way: its edge set consists of $E(G^*)$, along with at most $\Delta n / 200$ other edges.

Considering that we can describe every $D \in \Gamma(G^*)$ in such a way, we deduce that $|\Gamma(G^*)|$ is at most the number of ways of choosing a subset of at most $\Delta n/200$ edges from the collection of $\binom{n}{2}$ pairs.
Hence
\[
  |\Gamma(G^*)| \leq \textstyle\binom{n}{2}^{\Delta n/200} \leq n^{\Delta n / 100} = 2^{\Delta n\log n/100}.
\]
This contradicts \eqref{eq:GammaLB}.
\end{proof}

\section{Approximation with Steiner Nodes}\label{sec:steiner-apx}

In this section, we give a constant-factor approximation algorithm for the bounded network design with Steiner nodes. Let $\Delta$ be the degree bound, $V = \{ 1, 2, \dotsc, n  \}$ the set of nodes and $p \colon \binom{V}{2} \to [0,1]$ the probability distribution given as input. Let $m$ denote the number of demand edges.


\paragraph{Approximation algorithm.} We now present the approximation algorithm itself. The algorithm constructs the host graph in two simple steps: first, for each node $v \in V$ separately, we construct a Huffman tree for the distribution $p_v$, i.e., the conditional distribution of the other endpoint given that $v$ is included in the demand edge. Second, we connect these trees together to obtain the final host graph.

In more detail, the algorithm constructs the host graph $G$ as follows:
\begin{enumerate}
	\item For each node $v \in V$, let $T_v$ be the $(\Delta-1)$-ary Huffman tree for probability distribution $p_v$, ignoring any probabilities of $0$. Denote the leaf node corresponding to probability $p_v(u)$ as $t_{v,u}$, and let $d_{T_v}(v, t_{v,u})$ be the distance from $v$ to $t_{v,u}$ in $T_v$.
	
		  In the special case that $p_v(u) = 1$ for some $u \in V$, we take $T_v$ to be a tree with a root and a single leaf instead of the trivial tree.
		  
	\item For each pair $\{ u, v \} \subseteq V$ with $p(\{u,v\}) \ne 0$, we add an edge between the parent of $t_{v,u}$ in $T_v$ and the parent of $t_{u,v}$ in $T_u$, and remove the nodes $t_{v,u}$ and $t_{u,v}$.
\end{enumerate}
See Figure~\ref{fig:steiner-node-apx} for illustration.
\begin{figure}
\begin{center}
\includegraphics[width=0.8\textwidth]{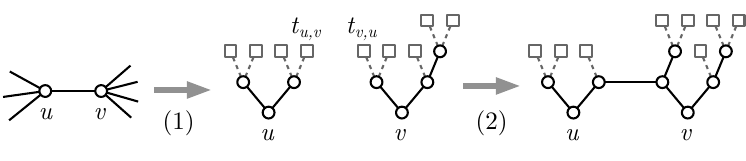}
\end{center}
\caption{ Approximation algorithm for bounded network design with Steiner nodes. In Step~(1), a Huffman tree is constructed for all nodes; the leaf nodes marked by squares are not included in the final host graph. In Step~(2), the corresponding leaf nodes are deleted and replaced by an edge.}\label{fig:steiner-node-apx}
\end{figure}

\paragraph{Time complexity.} To analyze the time complexity of the approximation algorithm, we first observe that constructing the Huffman tree in Step~(1) for node $v$ takes $O(d_v \log d_v)$ time, where $d_v = \deg_D(v)$ is the degree of $v$ in the demand graph. We have
\[ \sum_{v \in V} d_v \log_2 d_v \le \log_2 m \sum_{v \in V} d_v = 2 m \log_2 m\,,\]
and thus constructing Huffman trees for all nodes takes $O(m \log n)$ time in total. Step~(2) can clearly also be implemented in $O(m\log n)$ time, so the total time complexity is $O(m \log n)$.

\paragraph{Cost of solution.} Consider a graph $G$ as constructed above. By construction, we have
\[ d_G(u,v) = d_{T_v}(v, t_{v,u}) + d_{T_u}(u, t_{u,v}) - 1 \]
for all $u,v \in V$. Thus, the expected path length between nodes in $V$ in $G$ is
\begin{align*}
  & \sum_{\{u,v\} \in \binom{V}{2}} p(\{u,v\}) \bigl( d_{T_v}(v, t_{v,u}) + d_{T_u}(u, t_{u,v}) - 1 \bigr) 	\\
=\ & \sum_{v \in V} \sum_{u \in V\setminus \{ v \}} p(\{u,v\}) d_{T_v}(v, t_{v,u})  - \sum_{\{u,v\} \in \binom{V}{2}} p(\{u,v\}) && \text{(Rearrange)}  \\
=\ & \sum_{v \in V} p(v) \sum_{u \in V} p_v(u) d_{T_v}(v, t_{v,u}) - 1								&& \text{($p(\{ u,v\}) = p(v)p_v(u)$)}\\
\le\ & \sum_{v \in V} p(v) \Bigl( H_{\Delta-1}(p_v) + 1 \Bigr) - 1									&& \text{(Lemma~\ref{lemma:huffman})}\\
=\ & \sum_{v \in V} p(v) H_{\Delta-1}(p_v) + \sum_{v\in V }p(v) - 1\\
=\ & \sum_{v \in V} p(v) H_{\Delta-1}(p_v) + 1\,.													&& \text{($\sum_{ v \in V } p(v) = 2$)} \\
:=\ & C\,.
\end{align*}

By Theorem~\ref{thm:lower-bound}, the expected path length $C^*$ of an optimal solution is $\frac{1}{2}\sum_{v \in V} p(v) H_{\Delta+1}(p_v) - 1$ or greater. Comparing this to the expected path length upper bound $C$ on $G$, we have
\begin{align*}
C & = \sum_{v \in V} p(v) H_{\Delta-1}(p_v) + 1	\\
& = \frac{2\log_2(\Delta+1)}{\log_2(\Delta-1)}\Bigl(\frac{1}{2}\sum_{v \in V} p(v) H_{\Delta+1}(p_v) - 1\Bigr) + 1 + \frac{2\log_2(\Delta+1)}{\log_2(\Delta-1)}\\
& \le \frac{2\log_2(\Delta+1)}{\log_2(\Delta-1)} C^* + 1 + \frac{2\log_2(\Delta+1)}{\log_2(\Delta-1)}\,.
\end{align*}
Thus, the approximation factor of the algorithm is $\frac{2\log_2(\Delta+1)}{\log_2(\Delta-1)} + o(1)$. Note that $\frac{2\log_2(\Delta+1)}{\log_2(\Delta-1)}$ is at most $4$ for all $\Delta \ge 3$, and at most $3$ for all $\Delta \ge 4$.

\paragraph{Number of nodes.} Let $n = |V|$ and let $m$ be the number of non-zero entries in $p$. 
We now want to analyze the total number of nodes in $V \cup U$ in the graph $G$ given by the above algorithm.

\begin{fact}[e.g.~\cite{2019information}]\label{fact:huffman-structure}
Given a $d$-ary Huffman tree $T$, the leaves can be rearranged so that all internal nodes of $T$ have exactly $d$ children, with the exception of one internal node located at the maximum distance from the root. This node has at least $2$ children.
\end{fact}

\begin{fact}\label{fact:tree-internal-nodes}
A full $d$-ary tree with $\ell$ leaves has $\frac{\ell - 1}{d - 1}$ internal nodes.
\end{fact}

Now consider the Huffman tree $T_v$ constructed in Step~(1) of the algorithm. Let $d = \Delta - 1$ be the arity of the Huffman trees. Denoting the outdegree of $v$ in $p$ by $\deg(v)$, and assuming $\deg(v) \ge 2$, we have that $T_v$ has $\deg(v)$ leaves. By Fact~\ref{fact:huffman-structure}, we can turn $T_v$ into a full $d$-ary tree by adding at most $d-2$ leaves, which has $\ell \le \deg(v) + d - 2$ leaves. By Fact~\ref{fact:tree-internal-nodes} it has at most
\[ \frac{\ell - 1}{d - 1} \le \frac{\deg(v) - 1 + d - 2}{d-1} = \frac{\deg(v) - 2}{d-1} + \frac{d-1}{d-1} = \frac{\deg(v) - 2}{d-1} + 1 \]
internal nodes; this is also an upper bound for the number of internal nodes in $T_v$.

In the special case $\deg(v) = 1$, the above bound does not hold. However, we can use a slightly looser bound of $\frac{\deg(v) - 1}{d-1} + 1$ to cover all cases.

Since the leaf nodes of the Huffman trees $T_v$ are deleted in Step~2 of the algorithm, the total number of nodes in $G$ is the total number of internal nodes in these trees -- note that the roots of the trees $T_v$ correspond to the original nodes in $V$.
Thus, the number of nodes in $G$ is at most
\begin{equation}\label{eq:node-insertion-node-count}
\sum_{v \in V}\Bigl( \frac{\deg(v)-1}{d-1} + 1 \Bigr) = \Bigl( 1 - \frac{1}{d-1} \Bigr)n + \sum_{v \in V} \frac{\deg(v)}{d-1}  \le \Bigl( 1 - \frac{1}{d-1} \Bigr)n + \frac{2}{d-1}m\,.
\end{equation}
If there are no nodes with $\deg(v) = 1$, then the above bound improves to $\bigl( 1 - \frac{2}{d-1} \bigr)n + \frac{2}{d-1}m$.


\section{Applications}\label{sec:demand-balancing-model}\label{sec:demand-balancing-alg}

This section discusses applications of our approach.
First, we will use it to derive improved approximation bounds for the original
demand-aware network design problem which results in a constant factor approximation for the case where the maximum degree is close the average degree of the demand graph.
Second, we show how it can be used to obtain efficient heuristics.

\subsection{Improved approximation algorithm for demand balancing}\label{sec:demand-balancing-algstart}
In this section, we use the ideas from our approximation algorithm for bounded network design with Steiner nodes to design a \emph{demand balancing} algorithm in the spirit of the results of Avin et al.~\cite{avin2020demand}. That is, we consider a demand graph with average degree $\avgdeg{}$, but with possibly some nodes with very high degree, and we want to design a host graph \emph{without Steiner nodes} that has maximum degree close to $\avgdeg{}$.

Here, we give an improved version of Theorem~4 of \cite{avin2020demand}, showing how to construct a host graph with maximum degree $4\avgdeg{}+1$ and better expected path length guarantees. Specifically, we improve the maximum degree guarantee from $12\avgdeg{}$ to $4\avgdeg{} + 1$ and the expected path length guarantee by a factor of $\log_2(2\avgdeg{})$, which results in a constant factor approximation.

\subsubsection{Algorithm} \janne{Only place in the paper with subsubsections, I would prefer not having this many layers of nesting}

For technical convenience, we assume $\avgdeg{}$ is an integer and $n$ is divisible by $2$. We say that a node $v$ is a \emph{high-degree node} if its degree $\deg(v)$ in the demand graph satisfies $\deg(v) > 2\avgdeg{}$, and otherwise we say it is a \emph{low-degree node}. By Markov's inequality, the number of high-degree nodes is at most $n/2$.

For each node $v$, we say that the up to $2\avgdeg{}$ demand graph edges $\{ u, v \}$ for which $p_v(u)$ are highest are the \emph{high-demand edges for $v$}, and the rest of the edges are \emph{low-demand edges for $v$}. Note that only heavy nodes have low-demand edges.

We now construct the host graph $G = (V,E)$ in two steps as follows.

\paragraph{Step 1: Low-demand trees.} For a heavy node $v$, let $L_v$ and $H_v$ be the set of neighbours $u$ of $v$ in the demand graph such that $\{u, v\}$ is low-demand and high-demand, respectively.
Define a distribution~$q_v$ on $L_v$ by selecting an arbitrary node $s \in L_v$ and setting
\[ q_v(u) = \begin{cases}
p_v(u) + \sum_{w \in H_v} p_v(w) &  \text{if $u = s$, and}\\
p_v(u) & \text{otherwise.}
\end{cases} \]
As in Section~\ref{sec:steiner-apx}, we construct a $2\avgdeg{}$-ary Huffman tree for terminal set $L_v$ wrt. the probabilities~$q_v(u)$. As before, denote the leaf node corresponding to probability $u$ as $t_{v,u}$, and let~$d_{T_v}(v, t_{v,u})$ be the distance from $v$ to $t_{v,u}$ in $T_v$.

We repeat the construction for all heavy nodes, and map the internal nodes of the trees injectively to nodes of $V$ so that root of the tree for $v$ is $v$. We then add all edges of the tree to the routing graph $G$. This is always possible, since by applying \eqref{eq:node-insertion-node-count} with $N = n/2$, $M = \avgdeg{} n /2$, and $d = 2\avgdeg{}$, we have that there are at most
\[
\Bigl( 1 - \frac{1}{2\avgdeg{} - 1} \Bigr)\frac{n}{2} + \frac{2}{2\avgdeg{} - 1}\frac{\avgdeg{} n}{2}= \Bigl( 1 + \frac{ 2 \avgdeg{}  - 1 }{2 \avgdeg{} - 1}\Bigr) \frac{n}{2} \le n
\]
internal nodes, including the roots, in the constructed trees. 

\paragraph{Step 2: Connecting edges.} Now consider edge $\{ u, v \}$ in the demand graph. There are three cases to consider:
\begin{enumerate}
	\item Edge $\{ u, v \}$ is high-demand for both $u$ and $v$. In this case, we add edge $\{ u, v \}$ to $G$.
	\item Edge $\{ u, v \}$ is low-demand for both $u$ and $v$. In this case, we add an edge from the parent node of $t_{v,u}$ to the parent node of $t_{u,v}$ to $G$.
	\item Otherwise, without loss of generality, edge $\{ u, v \}$ is high-demand for $u$ and low-demand for $v$. In this case, we add an edge from $u$ to the parent of $t_{v,u}$ to $G$.
\end{enumerate} 

\subsubsection{Analysis}

\paragraph{Time complexity.} By the same argument as in Section~\ref{sec:steiner-apx}, we can see that computing the high-demand edges and building the Huffman trees for low-demand edges can be done in time $O(m \log n)$ in total. All other steps of the algorithm can be implemented with at most $O(m \log n)$ time.

\paragraph{Maximum degree.} Low-demand trees constructed in Step 1 have arity $2\avgdeg{}$, and thus contribute at most $2\avgdeg{} + 1$ to the degree of each node. This includes edges added for low-demand edges. Moreover, each node has at most $2\avgdeg{}$ edges added for handling high-demand edges, for a total of $4 \avgdeg{} + 1$.

\paragraph{Expected path length.} First, let us recall the grouping rule for entropy.

\begin{lemma}\label{lemma:entropy-grouping}
Let $p\colon \{ 1, 2, \dotsc, m \}$ be a distribution, and let $q \colon \{ 1,2,\dotsc, m-1 \}$ be a distribution defined by $q(i) = p(i)$ for $i = 1,2,\dotsc, m-1$, and $q(m-1) = p(m-1) + p(m)$. Then 
\[ H_d(p) = H_d(q) + (p(m-1) + p(m)) H_d\Bigl( \frac{p(m-1)}{p(m-1) + p(m)}, \frac{p(m)}{p(m-1) + p(m)} \Bigr)\,. \]
\end{lemma}

It follows immediately from Lemma~\ref{lemma:entropy-grouping} that for a heavy node $v$ we have  $H_d(q_v) \le H_d(p_v)$ for any $d \ge 2$.

Now consider a non-zero demand $\{u,v\}$. There are four cases to consider for the distance between $u$ and $v$ in $G$:
\begin{enumerate}
	\item $u \in H_v$ and $v \in H_u$: distance is $d_G(u,v) = 1$.
	\item $u \in H_v$ and $v \in L_u$: distance is $d_G(u,v) = d_{T_u}(u,t_{u,v})$.
	\item $u \in L_v$ and $v \in H_u$: distance is $d_G(u,v) = d_{T_v}(v,t_{v,u})$.
	\item $u \in L_v$ and $v \in L_u$: distance is $d_G(u,v) = d_{T_v}(v,t_{v,u}) + d_{T_u}(u,t_{u,v}) - 1$.
\end{enumerate}

Let $D$ be the set of demand edges, i.e., edges with non-zero demand. For a predicate $P$, denote by $[P]$ a function that is $[P] = 1$ if predicate $P$ is true, and $[P] = 0$ otherwise. The expected path length in $G$ is now
\begin{multline*}
\sum_{\{u,v\} \in D} p(\{u,v\}) \Bigl(  [u \in H_v][v \in H_u] + [u \in H_v][v \in L_u] d_{T_u}(u,t_{u,v}) + [u \in L_v][v \in H_u] d_{T_v}(v,t_{v,u}) \\
								   + [u \in L_v][v \in L_u] (d_{T_v}(v,t_{v,u}) + d_{T_u}(u,t_{u,v}) - 1)\Bigr)\,.
\end{multline*}
By re-grouping terms, this is equal to

\begin{align*}
\sum_{\{u,v\} \in D} p(\{u,v\}) \Bigl( & d_{T_v}(v,t_{v,u}) \bigl( [u \in H_v][v \in L_u] + [u \in L_v][v \in H_u] \bigr) \\
										& + d_{T_u}(u,t_{u,v}) \bigl( [u \in H_v][v \in L_u]  + [u \in L_v][v \in L_u] \bigr) \\
										& + [u \in H_v][v \in H_u] [u \in L_v][v \in L_u] \Bigr)\,.
\end{align*}		
%
By omitting the last negative term, and rearranging further, this is at most
\begin{align*}
& \sum_{u \in V} \sum_{v \in L_u \cup H_u} p(\{u,v\}) \Bigl(d_{T_u}(u,t_{u,v}) \bigl( [u \in H_v][v \in L_u]  + [u \in L_v][v \in L_u] \bigr)
+ [u \in H_v][v \in H_u] \Bigr)\\
\le  & \sum_{u \in V} \Bigl( \sum_{v \in L_u} p(\{u,v\}) d_{T_u}(u,t_{u,v}) + \sum_{v \in H_u} p(\{u,v\}) \Bigr)\\
\le &  \sum_{u \in V} p(u) \Bigl( \sum_{v \in L_u} p_u(v) d_{T_u}(u,t_{u,v}) + \sum_{v \in H_u} p_u(v) \Bigr)\,.
\end{align*}

For fixed $u \in V$, we have
\begin{align*}
\sum_{v \in L_u} p_u(v) d_{T_u}(u,t_{u,v}) + \sum_{v \in H_u} p_u(v) & \le \sum_{v \in L_u} p_u(v) d_{T_u}(u,t_{u,v}) + \sum_{v \in H_u} p_u(v) d_{T_u}(u,t_{u,s}) \\
																     & = \sum_{v \in L_u \setminus \{ s \}} p_u(v) d_{T_u}(u,t_{u,v}) + \bigl(p_u(s) + \sum_{v \in H_u} p_u(v) \bigr) d_{T_u}(u,t_{u,s})\\
																	 & = \sum_{v \in L_u} q_u(v) d_{T_u}(u,t_{u,v})\\
																	 & \le H_{2\avgdeg{}}(q_u) + 1 \le H_{2\avgdeg{}}(p_u) + 1\,.
\end{align*}
Thus, the expected path length in $G$ is at most
\[ \sum_{u \in V} p(u) H_{2\avgdeg{}}(p_u) + 1\,.\]
In comparison, by Theorem~\ref{thm:lower-bound} the expected path length for degree-$(4\avgdeg{}+1)$ host graphs is at least $\frac{1}{2}\sum_{u \in V} p(u) H_{4\avgdeg{}+2}(p_u)$, meaning that $G$ is a constant-factor approximation for the optimal degree-$(4\avgdeg+1)$ host graph.

\subsection{Optimized threshold for demand balancing}\label{sec:opt-thresh-alg}
From the demand balancing algorithm in Section~\ref{sec:demand-balancing-alg} we notice that there is one crucial value that appears multiple times. This value is $2\avgdeg{}$ which defines the threshold for high- and low-degree vertices, high and low-demand edges, and the arity of the Huffman trees.

In the following we instead define a new threshold value $t \in \mathbb{N}, t \geq 2$ and adapt the demand balancing algorithm to it. Crucially, for low thresholds $t$ the algorithm may fail.

A node is called \emph{high-degree} if $\deg(v) > t$, otherwise it is \emph{low-degree}. The first up to $t$ highest weight edges incident to a vertex $v$ are called \emph{high-demand edges for $v$}, and the remaining are are \emph{low-demand edges for $v$}.

We now execute Steps 1 and 2 of the demand balancing algorithm using these adapted definitions, and using $t$-ary Huffman trees in Step 1. Observe that if $t$ is chosen too small then the number of internal nodes in the Huffman trees will be larger than the number of low degree vertices and the algorithm will fail to construct a mapping.

Depending on $t$ the computed host graph will have a different maximum degree and expected path length. Specifically, the maximum degree is at most $2t+1$ and the expected path length is at most $\sum_{u \in V} p(u) H_{t}(p_u) + 1$. Assuming $t < 2\avgdeg{}$ then compared to the demand balancing algorithm this gives a lower maximum degree, but at the cost of a worse expected path length guarantee. However, the expected path length upper bound is only worse by a factor of $\log_t(2\avgdeg{})$.

From Section~\ref{sec:demand-balancing-alg} we know that for $t = 2\avgdeg{}$ the adapted algorithm is simply the demand balancing algorithm and can not fail. We will call the \emph{optimized threshold balancing} algorithm the output of the modified demand balancing algorithm with the smallest possible threshold value $t$. This value can be efficiently computed using binary search.

\subsection{Fixed-degree heuristic}\label{sec:fixed-degree-heuristic}

The main challenge of our methods, as well as the prior work, is that we do not have approximation algorithms for bounded network design that strictly respect a given arbitrary degree bound, without using additional resources such as Steiner nodes. However, in practice we expect that parameters such as the degree bound and number of Steiner nodes are limited by the physical properties of the network infrastructure, so solutions respecting these limitations would be requiring. 

In this section, we give a non-trivial heuristic algorithm for constructing a degree-$\Delta$ host graph for a given demand graph. We discuss its experimental performance in terms of expected path length in Section~\ref{sec:experiments}.

\paragraph{Algorithm overview.}

We first give a high-level overview of the main ideas of the algorithm, and defer some specific implementation detail choices to later in this section.

Let $\Delta \ge 6$ be the maximum degree desired, and choose parameters $d_1, d_2 \geq 3$ such that $d_1 + d_2 = \Delta$. Note that this choice can be made in multiple ways; we discuss the best practical solution later on.

The algorithm first splits the edge set of the demand graph into \emph{heavy edges} and \emph{light edges}. For the heavy edges, we construct a near-optimal solution of degree $d_1$, while for the light edges we use a trivial solution of degree $d_2$. In more detail:
\begin{enumerate}
	\item The heavy edges are selected from the highest-demand edges so that we can apply the approximation algorithm for bounded network design with Steiner nodes with degree bound $d_1$ to the graph induced by these edges using at most $n$ total nodes.
	\item The remaining edges are light edges. To cover the light edges, we use an arbitrary connected graph of maximum degree-$d_2$ to ensure that all pairs of nodes are connected by the host graph.
\end{enumerate}
We then overlay the solutions to both instances on top of the original node set.
In particular, the resulting host graph will have degree $d_1 + d_2 = \Delta$ by construction.

\paragraph{Heavy edges.} Consider a subset of demand graph edges $\hat{E}$, and the induced subgraph $G[\hat{E}] = ( \hat{V}, \hat{E} )$, where $\hat{V} = V \cap \bigcup_{e\in \hat{E}} e$. We can define a new bounded network design instance on $\hat{V}$ by setting
\[\hat{p}(\{u,v\}) = \frac{p(\{u,v\})}{\sum_{u,v\in \hat{V}} p(\{u,v\})}\,.\]
Define $\hat{p}(v)$ and $\hat{p}_v(u)$ analogously to $p(v)$ and $p_v(u)$, respectively. Note that we have
\[
\hat{p}(\{u,v\}) \ge p(\{u,v\})\,, \hspace{20mm}  \hat{p}_v(u)  \ge p_v(u)
\]
for all $u,v \in \hat{V}$.

To select the set of heavy edges $\hat{E}$, we sort all the edges by their demand in descending order, and select a prefix of this list so that when the approximation algorithm for bounded network design with Steiner nodes is applied to the instance $\hat{G} = ( \hat{V}, \hat{E} )$ with degree bound $d_1$, the resulting host graph has at most $n$ nodes in total. We can then overlay this solution on $V$ by identifying the Steiner nodes with nodes in $V \setminus \hat{V}$ in an arbitrary way.

Concrete bounds for the length of the selected prefix can be derived from the analysis of the approximation algorithm, but we found that in practice the best solution is to use binary search and to run the approximation algorithm to exactly compute the size of the host graph for any given prefix of the edge list. This adds an $O(\log n)$ factor to the running time, but this did not cause relevant slow-down in the experiments.

\paragraph{Light edges.} To cover the light edges, we can use any universal solution that guarantees an optimal girth of $\Theta(\log_{d_2} n)$. The simplest option is to use an arbitrary $(d_2-1)$-ary tree of height at most $\log_{d_2-1} n + 1$ spanning $V$, and add it to the output graph. However, we found that in practice a random $d_2$-regular graph is a better solution -- note that random regular graphs are known to be good expanders.

One may notice that the graph resulting from the \steinernodealg{} algorithm on the heavy edges may also have many vertices with degree less than $d_1$. In this case it would be more beneficial to randomly insert edges into this graph as long as the maximum degree is less than $\Delta$, instead of overlaying a graph with maximum degree $d_2$, as this might allow for more edges to be inserted.

\subsection{Miscellaneous fixed degree heuristics}\label{sec:misc-heuristics}
We are not aware of any heuristics in the literature that efficiently solve the Bounded Network Design problem for arbitrary $\Delta$.
Here we propose multiple simple heuristics to which we will ultimately compare the fixed degree heuristic in Section~\ref{sec:fixed-degree-heuristic} to in the experiments in Section~\ref{sec:experiments}.

\paragraph{Random tree.}
A simple approach is to simply generate a random $(\Delta-1)$-ary tree, by for example randomly permuting the vertices of the demand graph and building the tree layer by layer.
This is very quick to compute, but clearly it is oblivious to the actual demand graph.
What is more the leaves of the tree are poorly utilized since they have degree 1 and could therefore be incident to many more edges.
The leaves of the tree make up the majority of vertices if the tree is full.

\paragraph{Random graph.}
A desirable property of a DAN algorithm would be to output a graph in which all vertices have degree exactly $\Delta$, that is, a regular graph. This way the expected path length cannot be improved by adding more edges. For generating random $\Delta$-regular graphs the the approach by \citet{kim2003} is commonly used.
It generates asymptotically uniform regular graphs for $\Delta \in O(n^{1/3 - \epsilon})$ in $O(n\Delta^2)$ time.
Of course, a necessary condition for a $\Delta$-regular graph on $n$ vertices to exist is that $n\Delta$ is even, which in our case may not be guaranteed.

The algorithm by \citet{kim2003}, in much simplification, iteratively adds random edges to an empty graph until every vertex has degree $\Delta$.
This may fail even if $n\Delta$ is even, and then the algorithm is restarted.
Since in general $n\Delta$ may not be even and a regular graph is not guaranteed to exist, we propose to simply run the algorithm by \citet{kim2003} and if it fails to produce a $\Delta$-regular graph to then simply output the graph it has computed until the point where the algorithm fails.
Such a graph is maximal in the sense that no further edges can be added without creating vertices with degree greater than $\Delta$ and should therefore be nearly regular.

A downside to this approach is that, similar to the random tree heuristic, it is oblivious to the demand graph, and furthermore may not even return a connected graph. In the case that the returned graph is not connected one may want to sample a different random graph.

\paragraph{Greedy edge selection.}
Another possible approach is to start with an empty graph and iteratively add the highest weight edges from the demand graph as long as it does not increase the degree of a vertex above $\Delta$.
This way we focus on the heaviest edges first which have a big impact on the expected path length.
However, there is no guarantee that the graph constructed in this way will be connected.
Consider for example a graph which consists of a clique on 100 vertices and one other vertex which is connected to one vertex of the clique. The edges in the clique have higher weight than the single edge outside the clique. If $\Delta < 100$ then the algorithm will first pick edges of the clique and will then be unable to add the smallest weight edge which is the edge that connects the single vertex outside the clique. In this case the expected path length will be infinite.

\paragraph{Greedy edge deletion.}
To avoid the above problem of the greedy edge selection heuristic one may also iteratively delete edges from the graph if it is incident to vertices with degree greater than $\Delta$, starting with the lowest weight edges, but not deleting an edge if it would disconnect the graph.
This way we focus on connectivity of the resulting graph, as well as keeping many high weight edges in the graph, which have a big impact on the expected path length.
Unfortunately, this approach can also fail.
Consider for example a star with $\Delta+1$ leaves. The algorithm will be unable to delete any edge without disconnecting the graph, but the maximum degree is $\Delta+1$.

Note that this algorithm has $\mathcal{O}(m^2)$ worst case running time assuming $m>n$, since checking if deleting an edge disconnects the graph takes $\mathcal{O}(m)$ time. In practise this can be sped up significantly by maintaining a spanning tree of the graph and whenever we try to delete an edge that is not in the spanning tree we know that the connectivity of the graph is unaffected. Whenever an edge from the spanning tree is deleted and it does not disconnect the graph then the spanning tree can be patched using a single edge from the graph. Finding such an edge is often significantly faster than $\mathcal{O}(m)$.

\paragraph{Hybrid edge deletion.}
The greedy edge deletion approach is able to delete edges without disconnecting the graph, but it may not be able to obtain the desired maximum degree as described above. To alleviate this we propose to simply run the \fixeddegalg{} heuristic on the graph produced by the greedy edge deletion heuristic after normalizing the edge weights so that they sum to one. 


\section{Empirical Evaluation}\label{sec:experiments}

In this section we complement our theoretical insights with empirical evaluations using datacenter traces.

\subsection{Methodology} 

All presented algorithms in this paper were implemented and used in the experiments.
Furthermore, we provide comparisons to other methods such as the algorithm for sparse distributions from Theorem 4 by~\citet{avin2020demand}\footnote{The original algorithm is formulated for directed input graphs, whereas we assume the input to be undirected.
We split an undirected edge ${u,v}$ into two directed edges $(u,v)$, $(v,u)$ each receiving half of the undirected edge's probability.}, which we will refer to as the \helperalg{} algorithm, and a $k$-root graph approach by~\citet{peres-avin}, where $k = \log \Delta(D)$ which we will refer to as the \sqrtalg{} algorithm. For a full overview of the algorithms used in the experiments see Table~\ref{table:alg-overview}.

\begin{table}[t]
\caption{Overview of the algorithms used in the experiments. The Steiner nodes column refers to whether an algorithm solves the bounded network design problem with or without Steiner nodes.}\label{table:alg-overview}
\begin{tabular}{ c|c|c|c|c}
 \toprule
 Name & Shorthand & Ref. & Steiner nodes & $\Delta$ as input \\
 \midrule
 \steinernodealg{} & \steinernodealgshort{}         & Section~\ref{sec:steiner-apx}            & yes& yes\\
 \hline
 \demandbalancingalg{} & \demandbalancingalgshort{} & Section~\ref{sec:demand-balancing-algstart}   & no & no \\ 
 \optimizedthreshalg{} & \optimizedthreshalgshort{} & Section~\ref{sec:opt-thresh-alg}         & no & no \\ 
 \sqrtalg{} & \sqrtalgshort{}                       & Split\&Join \cite{peres-avin}            & no & no \\
\helperalg{}   &      \helperalgshort{}                      & Theorem 4 \cite{avin2020demand}          & no & no \\
 \hline
 \fixeddegalg{} & \fixeddegalgshort{}               & Section~\ref{sec:fixed-degree-heuristic} & no & yes\\
 \randtreealg{} & \randtreealgshort{}              & Section~\ref{sec:misc-heuristics}        & no & yes\\
 \randgraphalg{} & \randgraphalgshort{}             & Section~\ref{sec:misc-heuristics}        & no & yes\\
 \greedydeletealg{} & \greedydeletealgshort{}       & Section~\ref{sec:misc-heuristics}        & no & yes\\
 \greedyselectalg{} & \greedyselectalgshort{}       & Section~\ref{sec:misc-heuristics}        & no & yes\\
 \hybriddeletealg{} & \hybriddeletealgshort{}       & Section~\ref{sec:misc-heuristics}        & no & yes\\
 \bottomrule
\end{tabular}
\end{table}

We note that for the \hybriddeletealg{} algorithm we have substituted the \fixeddegalg{} heuristic for a slightly more unsafe version that may produce a disconnected graph, which however never happened in the experiments. This version gives the entire maximum degree to the \steinernodealg{} and uses a modified Huffman tree in which the root is allowed to have one additional child (since the root does not have a parent and thus lower degree). This has the benefit that if the \greedydeletealg{} heuristic already computes a graph with maximum degree $\Delta$ then the \steinernodealg{} will simply output an unmodified graph. For the \fixeddegalg{} heuristic we also use the improved second step in which edges are randomly inserted whenever it does not increase the maximum degree over $\Delta$, instead of overlaying a random graph with some maximum degree $d_2$.

Our implementation was written in C++ and compiled using g++ 12.2.1 with the \texttt{-O3} optimization flag.
All experiments were run on an AMD Ryzen 5 3600 CPU on a system with the Linux 6.2 kernel.
The computation time for the expected path length was not included in any running time measurements as it often vastly exceeded the running time of all algorithms discussed in this paper.
Furthermore, we note that the trick of maintaining a spanning tree in the \greedydeletealg{} heuristic reduced the running time from over a day for some instances to at most a few minutes.

\paragraph{Dataset.}
For the experiments we use the datacenter trace collection from \citet{avin2020complexity}. 
These traces contain temporal data, specifically source, destination and timestamps of communicating nodes.
To adapt these traces to our demand graph input format, we create and edge between two communicating nodes and set its weight equal to number of times the nodes communicated. The weights are normalized so that they can be interpreted as probabilities.
In Table~\ref{table:instance-stats} some selected parameters of these instances are summarized.

\begin{table}[t]
\centering
	\caption{Properties of demand graphs in the datacenter trace collection from \citet{avin2020complexity}. The table summarizes the number of nodes, edges, minimum degree, average degree, maximum degree, entropy and conditional entropy of the demand graphs. For the entropy calculations the base-2 logarithm was used.}\label{table:instance-stats}
\pgfplotstableread[col sep = comma]{data/instances.csv}\data 
\pgfplotstabletypeset[columns={instance,n,m,mindeg,avgdeg,maxdeg,entropy,cond_entropy},
	columns/instance/.style={string type,column name=Trace,column type = {r}},
	columns/n/.style={column name=$n$,precision=2,fixed},
	columns/m/.style={column name=$m$,precision=2,fixed},
	columns/mindeg/.style={column name=$\delta$,precision=2,fixed},
	columns/avgdeg/.style={column name=$\avgdeg{}$,precision=2,fixed},
	columns/maxdeg/.style={column name=$\Delta$,precision=2,fixed},
	columns/entropy/.style={column name=entropy,precision=2,fixed},
	columns/cond_entropy/.style={column name=cond. entropy,precision=2,fixed},
	every head row/.style ={before row=\toprule, after row=\midrule},
        every last row/.style ={after row=\bottomrule}
]{\data}
\end{table}

\subsection{Results}
We divide the results of the experiments into three groups, based on the algorithms which they involve, similar to the grouping in Table~\ref{table:alg-overview}.

\paragraph{Sparse heuristics.}

We first focus our attention on the bounded network design heuristics for which the maximum degree is a function of the input demand graph and cannot be specified by the user. These are the \sqrtalg{}, \helperalg{}, \demandbalancingalg{}, and \optimizedthreshalg{} heuristics.
In Table~\ref{table:heuristics-stats} properties of the resulting host graphs computed by these algorithms on our network trace dataset are presented.
Several observations can be made.
The \helperalg{} heuristic appeared to have the highest expected path length and second largest maximum degree compared to the other three approaches.
In contrast, the \demandbalancingalg{} heuristic achieved the lowest expected path length while also having the second lowest maximum degree, except for two demand graphs.
For four of the ten demand graphs the expected path lengths for the \demandbalancingalg{} and \sqrtalg{} heuristics were almost a third of the expected path length of the \helperalg{} heuristic.
The \sqrtalg{} heuristic achieved a low expected path length, but at the cost of having a comparatively large maximum degree, almost double that of the \demandbalancingalg{} heuristic for three demand graphs.
The \demandbalancingalg{} heuristic had generally the lowest expected path length and the second smallest maximum degree, whereas \optimizedthreshalg{} had generally the smallest maximum degree and the second smallest expected path length.
From this it can be seen that there is no clear best heuristic among the four considered algorithms. A trade-off has to be made between the expected path length and the maximum degree. Since in practise the maximum degree will often be limited by hardware, one might want to consider using the \optimizedthreshalg{} as the resulting demand-aware networks had less than half of the maximum degree compared to the other heuristics, which also means it had a maximum degree of less than $2\avgdeg{}(D)$, where $D$ is the demand graph. Additionally, the expected path length of \optimizedthreshalg{} is not much worse than that of \demandbalancingalg{}.

Curiously, for three instances the \sqrtalg{} heuristic computed a demand-aware network with a maximum degree equal or greater than the maximum degree in the demand graph, with an expected path length greater than one. In this case, it would have been possible to simply output the edges of the demand graph and achieve an expected path length of exactly one.

\begin{table}[t]
\centering
\caption{Properties of the host graphs computed by the \sqrtalg{}, \helperalg{}, \demandbalancingalg{} (\demandbalancingalgshort{}), and \optimizedthreshalg{} (\optimizedthreshalgshort{}) heuristics. We denote by $\Delta$ the maximum degree in the host graph and by EPL its expected path length. For each trace we highlight in gray the lowest maximum degree and in yellow the lowest expected path length, unless there is a tie.}\label{table:heuristics-stats}
\pgfplotstableread[col sep = comma]{data/heuristics.csv}\data 
\pgfplotstabletypeset[columns={instance,sqrt-maxdeg,sqrt-epl,avin-sparse-maxdeg,avin-sparse-epl,janne-sparse-maxdeg,janne-sparse-epl,janne-sparse-v2-maxdeg,janne-sparse-v2-epl},
	every head row/.style={
		before row={
			\toprule & \multicolumn{2}{c}{\sqrtalg{}} & \multicolumn{2}{c}{\helperalg{}} & \multicolumn{2}{c}{\demandbalancingalgshort{}} & \multicolumn{2}{c}{\optimizedthreshalgshort{}} \\
		},
		after row=\midrule,
	},
    every row 0 column janne-sparse-v2-maxdeg/.style={postproc cell content/.style={@cell content=\cellcolor{lightgray}{##1}}},
    every row 1 column janne-sparse-v2-maxdeg/.style={postproc cell content/.style={@cell content=\cellcolor{lightgray}{##1}}},
    every row 2 column janne-sparse-v2-maxdeg/.style={postproc cell content/.style={@cell content=\cellcolor{lightgray}{##1}}},
    every row 3 column janne-sparse-v2-maxdeg/.style={postproc cell content/.style={@cell content=\cellcolor{lightgray}{##1}}},
    every row 5 column avin-sparse-maxdeg/.style={postproc cell content/.style={@cell content=\cellcolor{lightgray}{##1}}},
    every row 7 column sqrt-maxdeg/.style={postproc cell content/.style={@cell content=\cellcolor{lightgray}{##1}}},
    %
    %
    every row 0 column janne-sparse-epl/.style={postproc cell content/.style={@cell content=\cellcolor{yellow}{##1}}},
    every row 1 column janne-sparse-epl/.style={postproc cell content/.style={@cell content=\cellcolor{yellow}{##1}}},
    every row 2 column janne-sparse-epl/.style={postproc cell content/.style={@cell content=\cellcolor{yellow}{##1}}},
    every row 3 column janne-sparse-epl/.style={postproc cell content/.style={@cell content=\cellcolor{yellow}{##1}}},
    every row 5 column janne-sparse-epl/.style={postproc cell content/.style={@cell content=\cellcolor{yellow}{##1}}},
	columns/instance/.style={string type,column name=Trace,column type = {r}},
	columns/avin-sparse-maxdeg/.style={column type=|r,column name=$\Delta$,precision=2,fixed},
	columns/avin-sparse-epl/.style={column name=EPL,precision=2,fixed},
	columns/janne-sparse-maxdeg/.style={column type=|r,column name=$\Delta$,precision=2,fixed},
	columns/janne-sparse-epl/.style={column name=EPL,precision=2,fixed},
	columns/sqrt-maxdeg/.style={column type=|r,column name=$\Delta$,precision=2,fixed},
	columns/sqrt-epl/.style={column name=EPL,precision=2,fixed},
 	columns/janne-sparse-v2-maxdeg/.style={column type=|r,column name=$\Delta$,precision=2,fixed},
	columns/janne-sparse-v2-epl/.style={column name=EPL,precision=2,fixed},
    every last row/.style ={after row=\bottomrule}
]{\data}
\end{table}

\paragraph{Fixed degree heuristics.}
In the above we have described the computational results for the heuristics is that the maximum degree is determined by the demand graph, and cannot be freely chosen.
We next compare different heuristics for which the maximum degree is an input parameter.
In Figure~\ref{fig:fixed-deg-average} the average expected path length over our whole dataset is plotted against the maximum degree desired for different heuristics.
It is apparent that the two heuristics that are oblivious to the demand graph, namely the \randtreealg{} and \randgraphalg{} heuristics have the worst expected path length, which was expected.
The \hybriddeletealg{} and \fixeddegalg{} heuristics appear to work equally well, with the \hybriddeletealg{} producing a slightly better output for small~$\Delta$, and slightly worse for bigger $\Delta$.

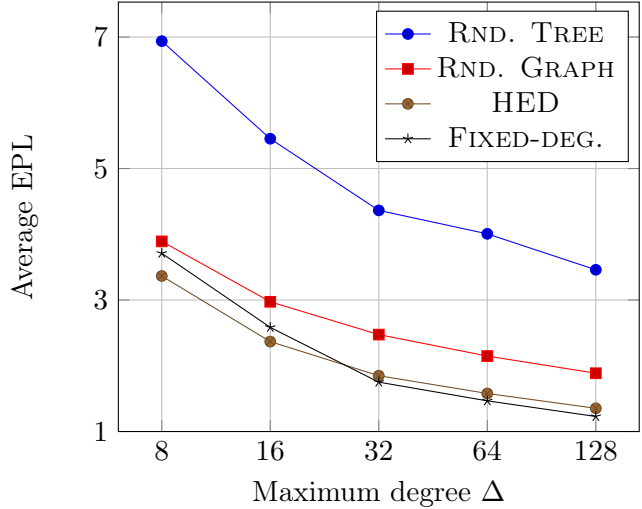
\begin{figure}[t]
 \centering
	\begin{tikzpicture}[scale=1]
		\begin{axis}[
			ylabel={Average EPL},
			xlabel={Maximum degree $\Delta$},
			grid,
			xmode=log,
			xtick={8,16,32,64,128},
            ytick={1,3,5,7},
			ymin=1,
			log ticks with fixed point
		]
  
		\addplot table[col sep=comma,y=rdat,x=b] {data/epl-vs-deg.csv};
		\addlegendentry{\randtreealgshort{}}
		\addplot table[col sep=comma,y=rdrg,x=b] {data/epl-vs-deg.csv};
		\addlegendentry{\randgraphalgshort{}}
      	\addplot table[col sep=comma,y=hybrid-delete,x=b] {data/epl-vs-deg.csv};
		\addlegendentry{\hybriddeletealgshort{}}
  		\addplot table[col sep=comma,y=fixed-deg,x=b] {data/epl-vs-deg.csv};
		\addlegendentry{\fixeddegalgshort{}}
		\end{axis}
	\end{tikzpicture}
	\caption{Average expected path length over the whole datacenter trace dataset depending on the maximum desired degree $\Delta$. The results for the \randtreealg{}, \randgraphalg{}, \hybriddeletealg{}, and \fixeddegalg{} heuristics are depicted. Not depicted are the results for the \greedydeletealg{} and \greedyselectalg{} heuristics as these sometimes failed. As some heuristics are randomized, each heuristic was ran 10 times on each input and its average expected path length was taken.}\label{fig:fixed-deg-average}
\end{figure}

Recall, that the fixed-degree heuristic is in fact a combination of two heuristics: a random graph based heuristic and one based on the \steinernodealg{} algorithm.
The maximum degree of the fixed-degree heuristic can be almost freely distributed among the two sub-heuristics.
In initial experiments we investigated ways to split the degree that results in the lowest expected path length.
There are three main ways: giving most of the degree to the \steinernodealg{} heuristics, or giving most of the maximum degree to the \randgraphalg{}, or splitting it evenly.

In our testing, giving the regular graph subheuristic a maximum degree of 3 and the remainder to the Steiner node based subheuristic appeared to provide the best results. This is also apparent from Figure~\ref{fig:fixed-deg-average}; the \randgraphalg{} heuristics has a higher expected path length than the \fixeddegalg{} heuristic with the configuration where the \steinernodealg{} algorithm gets most of the maximum degree. Therefore giving most of the maximum degree to the \steinernodealg{} heuristic is more beneficial for the expected path length.

Unsurprisingly, the \randtreealg{} heuristic is significantly worse than for example the \randgraphalg{} heuristic. The randomly constructed $(\Delta-1)$-ary tree has many leaves which by definition have a degree of one, and are thus poorly utilized as more edges could be inserted into the graph to improve the expected path length.

In Table~\ref{table:fixed-deg-heuristics-stats} we present results from Figure~\ref{fig:fixed-deg-average} for the case $\Delta = 32$ in more detail, and also include some results for heuristics which failed do solve some instances.
For the facebook instances the \fixeddegalg{} heuristic has the lowest expected path length. For some demands the heuristics achieve an expected path length of almost one, whereas the \textsc{Random Graph} and \textsc{Tree} heuristics have at least double the expected path length, which is not surprising as these two heuristics are oblivious to the input. The \greedydeletealg{} and \greedyselectalg{} fail on the facebook instances. Upon closer analysis of the structure of the facebook instances the reason becomes apparent. All three instances contain a central core of about 2\% of total vertices which almost form a clique. All remaining vertices are only connected to this core and do not communicate with each other. The edges within the core have a significantly higher weight compared to the edges outside of the core. The \greedyselectalg{} will then simply take edges within the core first and be then unable to take edges connecting the other vertices, making the graph disconnected. The \greedydeletealg{} heuristic fails as sometimes some vertices in the core have many degree-1 vertices adjacent to it, making it impossible to lower the maximum degree without disconnecting the graph.
Note that whenever the \greedydeletealg{} heuristic solves an instance then the expected path length is equal to that of the \hybriddeletealg{} heuristic, which is by design (recall the modification from the start of this Section).

\begin{table}[t]
\centering
\caption{Table summarizing the expected path length of the host graphs computed by the \fixeddegalg{}{}, \randtreealg{}, \randgraphalg{}, \greedydeletealg{}, \greedyselectalg{}, and \hybriddeletealg{} heuristics for $\Delta=32$. We denote by EPL the expected path length. If an algorithm failed on some instance then it is replaced with a "---" symbol. For each trace we highlight in bold the lowest expected path length, unless there is a tie.}\label{table:fixed-deg-heuristics-stats}
\pgfplotstableread[col sep = comma]{data/fixed-deg-heuristics.csv}\data 
\pgfplotstabletypeset[columns={instance,fixed-deg-32-epl,hybrid-delete-32-epl,rdrg-32-epl,rdat-32-epl,greedy-delete-32-epl,greedy-select-32-epl},
	every head row/.style={
		before row={
			\toprule \\
		},
		after row=\midrule,
	},
    every row 0 column fixed-deg-32-epl/.style={postproc cell content/.style={@cell content=\cellcolor{lightgray}{##1}}},
    every row 1 column fixed-deg-32-epl/.style={postproc cell content/.style={@cell content=\cellcolor{lightgray}{##1}}},
    every row 2 column fixed-deg-32-epl/.style={postproc cell content/.style={@cell content=\cellcolor{lightgray}{##1}}},
    every row 8 column greedy-select-32-epl/.style={postproc cell content/.style={@cell content=\cellcolor{lightgray}{##1}}},
    %
    %
	columns/instance/.style={string type,column name=Trace,column type = {r}},
	columns/fixed-deg-32-epl/.style={column type=|r,column name=$\fixeddegalgshort{}$,precision=2,fixed},
 	columns/hybrid-delete-32-epl/.style={column type=|r,column name=$\hybriddeletealgshort{}$,precision=2,fixed},
   	columns/rdrg-32-epl/.style={column type=|r,column name=$\randgraphalgshort{}$,precision=2,fixed},
    columns/rdat-32-epl/.style={column type=|r,column name=$\randtreealgshort{}$,precision=2,fixed},
    columns/greedy-delete-32-epl/.style={column type=|r,column name=$\greedydeletealgshort{}$,precision=2,fixed},
    columns/greedy-select-32-epl/.style={column type=|r,column name=$\greedyselectalgshort{}{}$,precision=2,fixed},
    every last row/.style ={after row=\bottomrule},
    empty cells with={---}
]{\data}
\end{table}

\paragraph{Steiner nodes.}
Our last experiments involve the \steinernodealg{} algorithm. In Table~\ref{table:steiner-stats} we also present computation results for the algorithm.
Recall, that the \steinernodealg{} algorithm is allowed to use additional nodes (Steiner nodes) that are not present in the demand graph.
For a relatively low maximum degree of eight the number of new nodes used by the algorithm is sometimes even greater than $30n$.
As the maximum degree increases, the expected path length decreases and so does the number of used Steiner nodes.
For some traces the number of Steiner nodes remains quite high. For example, for $\Delta=64$ on the second facebook trace the algorithm used $2.7n$ Steiner nodes.

\begin{table}[t]
\centering
\caption{Expected path length in the resulting host graphs computed by the \steinernodealg{} algorithm for different maximum degree bounds. We denote by $n'$ the number of nodes in the host graph and by $n$ the number of nodes in the demand graph.}\label{table:steiner-stats}
\pgfplotstableread[col sep = comma]{data/heuristics.csv}\data 
\pgfplotstabletypeset[columns={instance,ADDITIVE-8-nratio,ADDITIVE-8-epl,ADDITIVE-16-nratio,ADDITIVE-16-epl,ADDITIVE-32-nratio,ADDITIVE-32-epl,ADDITIVE-64-nratio,ADDITIVE-64-epl},
	every head row/.style={
		before row={
			\toprule & \multicolumn{2}{c}{$\Delta = 8$} & \multicolumn{2}{c}{$\Delta = 16$} & \multicolumn{2}{c}{$\Delta = 32$} & \multicolumn{2}{c}{$\Delta = 64$} \\ 
		},
		after row=\midrule,
	},
	columns/instance/.style={string type,column name=Trace,column type = {r}},
	columns/ADDITIVE-8-nratio/.style={column type=|r,column name=$n'/n$,precision=1,fixed},
	columns/ADDITIVE-8-epl/.style={column name=EPL,precision=2,fixed},
	columns/ADDITIVE-16-nratio/.style={column type=|r,column name=$n'/n$,precision=1,fixed},
	columns/ADDITIVE-16-epl/.style={column name=EPL,precision=2,fixed},
	columns/ADDITIVE-32-nratio/.style={column type=|r,column name=$n'/n$,precision=1,fixed},
	columns/ADDITIVE-32-epl/.style={column name=EPL,precision=2,fixed},
	columns/ADDITIVE-64-nratio/.style={column type=|r,column name=$n'/n$,precision=1,fixed},
	columns/ADDITIVE-64-epl/.style={column name=EPL,precision=2,fixed},
        every last row/.style ={after row=\bottomrule}
]{\data}
\end{table}

Furthermore, in Figure~\ref{fig:steiner-vs-fixed-deg} we compare the \steinernodealg{} algorithm to the \fixeddegalg{} heuristic which utilizes the \steinernodealg{} algorithm, but does not use Steiner nodes, that is it computes a host graph with the same vertex set as the demand graph. It can be observed that the \steinernodealg{} algorithm achieves a slightly lower expected path length on average, but as can be seen in Table~\ref{table:steiner-stats} this is at the cost of utilizing many Steiner nodes. The \fixeddegalg{} does not utilize Steiner nodes, but the expected path length is marginally larger.
Recall that the \fixeddegalg{} utilizes two heuristics and in the last part a random graph is overlaid to fill any gaps due to nodes with degree less than $\Delta$. This overlaid random graph is not present in the \steinernodealg{} algorithm.
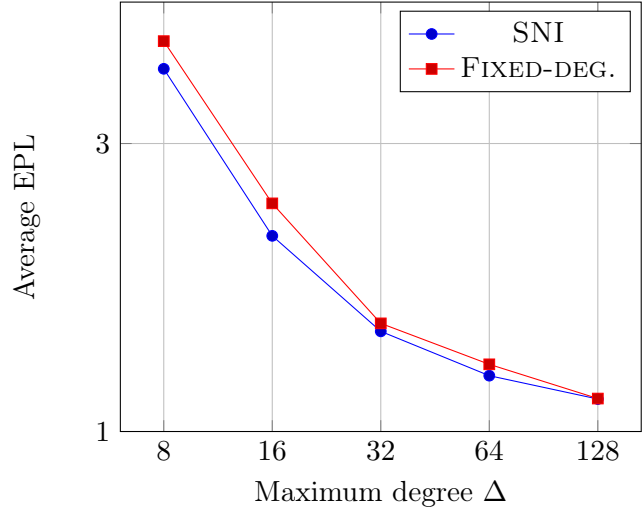
\begin{figure}
 \centering
	\begin{tikzpicture}[scale=1]
		\begin{axis}[
			ylabel={Average EPL},
			xlabel={Maximum degree $\Delta$},
			grid,
			xmode=log,
			xtick={8,16,32,64,128},
            ytick={1,3,5,7},
			ymin=1,
			log ticks with fixed point
		]
  
		\addplot table[col sep=comma,y=ADDITIVE,x=b] {data/epl-vs-deg.csv};
		\addlegendentry{\steinernodealgshort{}}
  		\addplot table[col sep=comma,y=fixed-deg,x=b] {data/epl-vs-deg.csv};
		\addlegendentry{\fixeddegalgshort{}}
		\end{axis}
	\end{tikzpicture}
	\caption{Average expected path length over the whole datacenter trace dataset depending on the maximum desired degree $\Delta$. The results for the \steinernodealg{} and \fixeddegalg{} heuristics are depicted.}\label{fig:steiner-vs-fixed-deg}
\end{figure}

\paragraph{Running times.}
Finally, we remark that the running times of all algorithms mentioned in this paper on the dataset were quite similar. Most algorithms needed 5--20 seconds for the facebook instances and at most $0.01s$ for the remaining ones. Only the \greedydeletealg{} and \hybriddeletealg{} heuristics had a relatively large running time of up to a few minutes on the facebook traces.

\section{Conclusion}\label{sec:conclusion}

This paper presented new approaches to designing bounded-degree demand-aware networks, leading to improved approximation algorithms for variants of the problem. We also reported on efficient implementations of our algorithms and presented various hardness results and lower bounds. In particular, we highlight our fixed-degree heuristic as a practical algorithm for bounded-degree network design problem; it is guaranteed to respect the given degree bound regardless of the demand distribution and, in experiments, achieves expected paths lengths close to a constant-factor approximation on real-world instances.

Our work leaves open several interesting directions for future research. In particular, 
it will be interesting to explore the tradeoff between approximation quality and resource augmentation. 
More generally, it remains open whether a constant-factor approximation algorithm exists for bounded-degree network design without the use of Steiner nodes. 

\bibliographystyle{plainnat}
\bibliography{network-design} 

\clearpage

\appendix

\section{Hardness of bounded-degree network design}\label{app:hardness}

\paragraph{Decision version of bounded-degree network design.} In this section, we show that the decision versions of bounded-degree network design, with or without Steiner nodes, is NP-complete for any fixed degree bound $\Delta \ge 2$. Formally, the decision version bounded of bounded-degree network design is defined as follows --  note that we consider an \emph{integer-weighted} version of the problem to sidestep considerations regarding encoding of small fractional number:

\begin{description}[noitemsep]
	\item[Instance:] A set $V$ with $|V| = n$, a weight function $w \colon \binom{V}{2} \to \mathbb{N}$, positive integers $\Delta$ and $K$.
	\item[Question:] Is the a graph $G = (V, E)$ such that maximum degree of $G$ is at most $\Delta$ and
	\[ \sum_{ \{u, v\} \in \binom{V}{2}} w(\{u,v\}) d_G(u,v) \le K\,, \]
where $d_G(u,v)$ denotes the distance between $u$ and $v$ in $G$.
\end{description}

The decision version of bounded-degree network design with Steiner nodes is defined analogously:

\begin{description}[noitemsep]
	\item[Instance:] A set $V$ with $|V| = n$, a weight function $w \colon \binom{V}{2} \to \mathbb{N}$, positive integers $\Delta$ and $K$.
	\item[Question:] Is there a graph $G = (V', E)$ such that $V \subseteq V'$, the maximum degree of $G$ is at most $\Delta$, and
	\[ \sum_{ \{u, v\} \in \binom{V}{2}} w(\{u,v\}) d_G(u,v) \le K\,, \]
where $d_G(u,v)$ denotes the distance between $u$ and $v$ in $G$?
\end{description}

While it is immediately clear that bounded-degree network design without Steiner nodes is in NP, it is not \emph{a priori} clear that the version with Steiner nodes always has optimal host graphs of polynomial size. However, not surprisingly this is indeed the case:

\begin{lemma}\label{lemma:steiner-node-bound}
For any instance of bounded-degree network design problem with Steiner nodes, there is a solution with optimal effective path length and at most a polynomial number of Steiner nodes.
\end{lemma}

\begin{proof}
For $\Delta = 2$, one can immediately observe that there is a solution with optimal effective path length without Steiner nodes, as we can replace any degree-$2$ Steiner node with an edge. Thus, assume $\Delta \ge 3$ and let $G = (V \cup U,E)$ be a host graph with optimal effective path length for an instance of bounded-degree network design with Steiner nodes. Moreover, assume that $G$ is minimal in the sense that any proper subgraph of $G$ has larger effective path length.

Let $u, v \in V$ be arbitrary pair of nodes with non-zero demand between them and distance $\ell \ge 4$ in $G$, and consider the shortest path $P = (u, v_1, v_2, \dotsc, v_{l-1}, v)$ between $u$ and $v$ in $G$. Let $s,t \in V$ be an another pair of nodes with non-zero demand between them, possibly including either $u$ or $v$ but not both. We say that an edge $e$ with exactly one endpoint on $P$ is \emph{critical} for $s$ and $t$ if removing it from $G$ would increase the distance between $s$ and $t$.

Consider a shortest path $P' = (s = u_0 , u_1, u_2, \dotsc, u_{k-1}, u_k = t)$ between $s$ and $t$ that intersects with $P$, and let $u_i$ and $u_j$ be the first and the last node on $P'$ that is also on $P$, respectively. As both $P$ and $P'$ are shortest paths, the distance between $u_i$ and $u_j$ must be the same along $P$ and $P'$. It follows that only $\{ u_{i-1}, u_i \}$ and $\{ u_j, u_{j+1} \}$ can be critical for $s$ and $t$, if they exist. Moreover, a critical edge must be on all shortest paths between $s$ and $t$, so there can be at most two critical edges for $s$ and $t$.

Now consider an arbitrary edge $e \in E$ with exactly one endpoint on $P$. By minimality of $G$, deleting $e$ would increase the distance between some pair $\{ s,t \} \subseteq V$; thus, edge $e$ is critical for some $s$ and $t$. By choice of $e$, it follows that all edges with exactly one endpoint on $G$ are critical for some $s$ and $t$. By above argument, there are at most $2\binom{n}{2} \le n^2$ critical edges, and thus at most that many edges with exactly one endpoint on $P$.

Assume for the sake of contradiction that $\ell \ge n^2 + 3$. Since there are at most $n^2$ edges with exactly one endpoint on $P$, there exist two non-adjacent nodes on $P$ that have degree less than $\Delta$. Adding an edge between these two nodes to $G$ would decrease the effective path length, which contradicts the optimality of $G$. Moreover, since $u$ and $v$ were chosen arbitrarily, it holds that for all pairs $u,v \in V$ with non-zero demand the distance between $u$ and $v$ in $G$ is $O(n^2)$. By minimality of $G$, it follows that $G$ has at most $O(n^4)$ nodes.
\end{proof}

\paragraph{\texorpdfstring{\boldmath  NP-completeness for $\Delta = 2$.}{NP-completeness for Δ = 2.}} For degree bound $\Delta = 2$ one can easily observe that, for both versions of the bounded-degree network problem, there is an optimal host graph that is a collection of cycles. Moreover, it is clear that Steiner nodes cannot improve the solution for $\Delta = 2$.

Avin et al.~\cite{avin2020demand} have noted that for \emph{connected} demand graphs, the bounded-degree network design problem for $\Delta = 2$ is very close to minimum linear arrangement. More precisely, we can reduce from the \emph{undirected circular arrangement problem}, known to be NP-complete~\cite{liberatore2002circular}, to prove that bounded-degree network design problems are NP-complete:
\begin{description}[noitemsep]
	\item[Instance:] Graph $G = (V,E)$ with $|V|= n$, a weight function $w \colon E \to \mathbb{N}$, a positive integer $K$.
	\item[Question:] Is there a bijection $f \colon V \to \{ 1, 2, \dotsc, n \}$ such that
		\[ \sum_{ e \in E} w(e) h(e) \le K\,, \]
		where $h(\{u, v\}) = \min\bigl( f(u) - f(v) \mod n, f(v) - f(u) \mod n )$\,.
\end{description}

The circular arrangement problem can be seen as an instance of bounded-degree network design where the host graph is required to be a cycle. In particular, circular arrangement is equivalent to bounded-degree network design with $\Delta = 2$ on connected instances, as in that case the host graph is also required to be connected -- i.e., a path or a cycle. Thus, the NP-completeness for both versions of bounded-degree network design follows from the following observation:

\begin{lemma}\label{lemma:ca-connected}
Undirected circular arrangement problem is NP-complete when restricted to connected graphs.
\end{lemma}

\begin{proof}
Membership in NP is trivial. For NP-hardness, we reduce from undirected circular arrangement on general inputs.

Let $G = (V,E)$, $w$ and $K$ be an instance of undirected circular arrangement.
We construct a new instance by letting $G' = (V, E')$ be a complete graph on $V$, and define the new edge weight function~$w'$ as
\[ w'(e) = \begin{cases}
	n^3 w(e) & \text{if $e \in E$, and }\\
	1        & \text{if $e \notin E$.}
\end{cases} \]
Finally, we set $K' = n^3 (K + 1) - 1$. Clearly $G'$ is a connected graph.

First consider the case where $G$ is a yes-instance, and let $f$ be a solution satisfying $\sum_{ e \in E} w(e) h(e) \le K$. Using $f$ as a solution for $G'$ and summing over the edges in $E$, we have that $\sum_{ e \in E} w'(e) h(e) \le n^3 K$. For edges not in $E$, we have
\[\sum_{ e \notin E} w'(e) h(e) \le \binom{n}{2} n < n^3\,.\]
Thus, 
\[\sum_{ e \in E'} w'(e) h(e) = \sum_{ e \in E} w'(e) h(e) + \sum_{ e \in E'} w'(e) h(e) < n^3 K + n^3\,,\]
which is at most $n^3 (K + 1) - 1$ as the value of the sum is an integer. Hence, $G'$ is a yes-instance.

For the other direction, assume that $G'$ a yes-instance, and let $f$ be a solution satisfying $\sum_{ e \in E'} w'(e) h(e) \le K'$. This implies that
\[ \sum_{ e \in E} w'(e) h(e) \le K' = n^3 (K+1) - 1\,.\]
Since all edge weights $w'(e)$ for $e \in E$ are divisible by $n^3$, the value of $\sum_{ e \in E} w'(e) h(e)$ is also divisible by $n^3$. This combined with above inequality implies
\[ \sum_{ e \in E} n^3 w(e) h(e) \le n^3 K\,,\]
and thus we have 
\[ \sum_{ e \in E} w(e) h(e) \le K\,.\]
Hence, $G$ is a yes-instance.
\end{proof}

\paragraph{\texorpdfstring{\boldmath  NP-completeness for $\Delta \ge 3$.}{NP-completeness for Δ ≥ 3.}} We prove NP-hardness for both variants of the bounded-degree network design problem. There are two main tricks we utilize in this proof, which we explain before proceeding with the proof proper.

First is the observation that when we construct instances for these problems, we can force any optimal host graph to have some specific edge $\{ u, v \}$ we desire by giving this pair a large demand $p(\{ u,v \}) = W$. If we then have $k$ edges we want to force, we can set the cost threshold of the instance to be between $kW$ and $kW-1$, so that any valid host graph will necessarily include all of these edges.

The second trick is used to ensure that nodes in our gadget constructions have the desired degrees. We use the following simple lemma:

\begin{lemma}\label{lemma:bnd-np-gadgets}
Let $d \ge 3$ be an odd integer. There is a graph $H_d = (V,E)$ such that $|V| = d+2$, one node in $H_d$ has degree $d-1$, and all other nodes in $H_d$ have degree $d$.
\end{lemma}

\begin{proof}
We start from a $(d+1)$-clique $K_{d+1}$, and then remove an arbitrary matching of size $(d-1)/2$. The resulting graph has $d-1$ nodes of degree $d-1$, and 2 nodes of degree $d$. We then add a new node and connect it with an edge to all the nodes of degree $d-1$. The resulting graph clearly has the claimed properties.
\end{proof}

We use Lemma~\ref{lemma:bnd-np-gadgets} to do what we will henceforth call \emph{attaching a degree-blocking gadget}: given a graph of odd maximum degree $\Delta$, some value $d \leq \Delta$, and a node $v$ with degree $\deg(v) < d$, we turn $v$ into a degree-$d$ node by adding $d-\deg(v)$ copies of $H_\Delta$ to the graph, and adding an edge between $v$ and degree-$(\Delta-1)$ node in each copy of $H_\Delta$. Clearly then $v$ will have degree $d$, and each added node will have degree exactly $\Delta$.

We now proceed to the main proof. For simplicity, we start with the case of odd maximum degree $\Delta$.

\begin{theorem}\label{thm:bnd-hardness-odd}
Bounded-degree network design and bounded-degree network design with Steiner nodes are NP-hard for any odd $\Delta \ge 3$.
\end{theorem}

\begin{proof}
We reduce from vertex cover on $3$-regular graphs; recall that the problem is known to remain NP-complete even with the degree restriction~\cite{garey-johnson, garey1976some}. More precisely, we will give a reduction that has the additional property that the resulting instance is a yes-instance without Steiner nodes if and only if it is a yes-instance with Steiner nodes, thus proving NP-hardness for both variants of the problem.

Let $(V,E,k)$ be an instance of vertex cover (in its decision form, where we ask if there is a vertex cover of size $k$). We first define the various gadget graphs we will use in the the construction.

\smallskip

\emph{Selector gadget.} Let $b = 2^{\lceil \log_2 k \rceil}$ be the smallest power of two larger than $k$. To construct the \emph{selector gadget}, we start with a complete binary tree with $b$ leaves and root node $r$. We arbitrarily select $k$ leaves as the \emph{selector nodes} $s_1, s_2, \dotsc, s_k$, and attach $\Delta-2$ degree-blocking gadgets $H_\Delta$ to each $s_i$. We also attach $\Delta-2$ degree-blocking gadgets $H_\Delta$ to the root $r$, and $\Delta-1$ degree-blocking gadgets to the remaining leaves. One can now verify that each selector node has degree $\Delta-1$, and the remaining nodes in the selector gadget have degree $\Delta$.

\smallskip

\emph{Vertex gadgets.} For each vertex $v \in V$ in the vertex cover instance, we construct a \emph{vertex gadget} as follows -- recall that by assumption $v$ has degree $3$. We take a complete binary tree with $4$ leaves, denoting the root of the tree by $r_v$. We identify $3$ of the leaves as \emph{terminal nodes $t_e$} for edges $e$ incident to node $v$. Finally, we attach $\Delta-3$ degree-blocking gadgets $H_\Delta$ to the root $r_v$, as well as $\Delta-2$ degree-blocking gadget to each terminal node and $\Delta-1$ degree-blocking gadgets to the remaining leaf. 

Finally, we combine the vertex gadgets for all vertices in $V$ by identifying two terminal nodes for edge $e = \{ u, v \}$ in the vertex gadgets of $u$ and $v$, as well as merging the respective degree-blocking gadgets into one. We can now observe that in the graph formed by the combined vertex gadget, the root nodes $r_v$ for $v \in V$ have degree $\Delta-1$, and all other nodes have degree $\Delta$.

\smallskip

\emph{Instance construction.} We now construct an instance $(V',w, K)$ of bounded-degree network design with Steiner nodes as follows. Let $V'$ be the set of nodes in the selector gadget and the combined vertex gadgets as constructed above. For any edge $\{ u, v \} \subseteq V'$ in the gadget constructions, we set the demand $w(\{ u, v \})$ to be $W = |E|(\log_2 b + 3) + 1$. For each edge $e \in E$ in the vertex cover instance, we add demand $w(\{ r, t_e \}) = 1$ between the root of the selector gadget and the terminal node for $e$. For all other pairs, we set the demand to $0$. Finally, we set the effective path length threshold for the bounded network design instance to be $K = MW + |E|(\log_2 b + 3)$, where $M$ is the total number of edges in the gadgets.

\smallskip

\emph{Correctness.} Now assume that $(V,E,k)$ is a yes-instance of vertex cover, and let $C = \{ v_1, v_2, \dotsc, v_k \} \subseteq V$ be a vertex cover. Let $G' = (V', E')$ be a host graph for instance $(V',w,K)$ obtained by taking all the gadget edges and edges $\{ s_i, r_{v_i} \}$ for $i = 1, 2, \dotsc, k$. The demands for the gadget edges are satisfied by $G'$ with direct edges, contributing a total of $MW$ to the effective path length. Now consider an arbitrary edge $e \in E$ and a demand between $r$ and $t_e$. Since $C$ is a vertex cover, there is a vertex $v_i \in C \cap e$ such that the edge $\{ s_i, r_{v_i} \}$ belongs to $E'$. Thus there is a path from $r$ to $t_e$ in $G'$ that starts from $r$ and travels along the selector gadget to $s_i$, crosses the edge $\{ s_i, r_{v_i} \}$ and travels from $r_{v_i}$ to $t_e$ along the vertex gadget for $v_i$. This path has length $\log_2 b + 3$. Since $e$ was selected arbitrarily and $C$ is a vertex cover, this holds for all edges $e \in E$. Therefore the total effective path length is $MW + |E|(\log_2 b + 3) = K$, that is, $(V',w,K)$ is a yes-instance.

For the other direction, assume that $(V',w,K)$ is a yes-instance, and let $G' = (V' \cup U,E')$ be a host graph with effective path length at most $K$. First, we observe that $G'$ must include all gadget edges, as otherwise the effective path length from the corresponding demands is at least $(M+1)W > K$. Now consider the path between $r$ and $t_e$ for an arbitrary $e \in E$. Since all gadget edges are included in $G'$, this path must exit the selector gadget via an edge connected to one of the selector nodes $s_i$, as all other nodes in the selector gadget have degree $\Delta$ from the edges inside the gadget. By the same argument, the path must enter the vertex gadget via one the nodes $r_v$ and travel via gadget edges to $t_e$. It follows that this path must have length at least $\log_2 b + 3$. Thus, for the effective path length to be at most $K = MW + |E|(\log_2 b + 3)$, the distance between $r$ and each $t_e$ must be exactly $\log_2 b + 3$.

We now observe that for an edge $e = \{ u, v \} \in E$, by the same reasoning as above, the only way for the path from $r$ to $t_e$ to be of length $\log_2 b + 3$ is that the path travels from $r$ to a selector node $s_i$, then from $s_i$ to either $r_v$ or $r_u$ by a direct edge, and then via the vertex gadget to $t_e$. Thus, for any edge $e = \{ u, v \} \in E$, either $\{ s_i, r_v \}$ or $\{ s_i, r_u \}$ is in $E'$ for some $i = 1, 2, \dotsc, k$. Thus, the set of vertices
\[ C = \bigl\{ v \in V \colon \{ s_i, r_v \} \in E' \text{ for some $i = 1, 2, \dotsc, k$}  \bigr\} \]
is a vertex cover of $G$, and since each $s_i$ may only be incident to one non-gadget edge due to the degree constraint, we have $|C| \le k$. Thus, $(V,E,k)$ is a yes-instance of vertex cover.
\end{proof}

We can now modify the above construction for even maximum degree as follows.

\begin{theorem}
Bounded-degree network design and bounded-degree network design with Steiner nodes are NP-hard for any even $\Delta \ge 4$.
\end{theorem}

\begin{proof}
Let $\Delta \ge 4$ be fixed and even. We start by taking the reduction given by Theorem~\ref{thm:bnd-hardness-odd} for $\Delta-1$, and modify it as follows. Let $N$ be the total number of nodes in the instance $(V',w,K)$ produced by the reduction, and let $M$ be the number of forced edges in the gadgets, i.e., the number of edges with demand $W$.

We now create the instance for degree $\Delta$ by taking two disjoint copies of the node set and demands of the degree-$(\Delta-1)$ instance $(V',w,K)$, and adding a demand $W$ for each pair of corresponding nodes from the two copies. All other demands for pairs from different copies is set to $0$. Finally, the effective path threshold is set to $2(MW + N)  + 2|E|(\log_2 b + 3)$.

If the vertex cover instance was a yes-instance, then the degree-$\Delta$ instance has a solution of cost $2(MW + |E|(\log_2 b + 3)) + NW$, obtained by taking the solution for degree-$(\Delta-1)$ instance for both copies, and adding all the edges between corresponding nodes in the two copies. On the other hand, if the degree-$\Delta$ instance is a yes-instance, then all high-weight edges between the two copies must be included in the solution, using up one edge for each node in the instance. Each copy of the degree-$(\Delta-1)$ instance must have effective path length of $MW + |E|(\log_2 b + 3)$ for its internal demands, which by same argument as in the proof of Theorem~\ref{thm:bnd-hardness-odd} is only possible if the original vertex cover instance is a yes-instance.
\end{proof}

Putting together all the results from this section, we obtain Theorem~\ref{thm:np-completeness-main}.

\end{document}